\documentclass[letterpaper, 10 pt, conference]{ieeeconf}
\IEEEoverridecommandlockouts
\usepackage{amsmath,amssymb,}
\usepackage{graphicx,epstopdf}
\usepackage{euscript,framed}
\newtheorem{theorem}{Theorem}[section]

\newtheorem{proposition}[theorem]{Proposition}

\newtheorem{definition}[theorem]{Definition}

\newtheorem{remark}[theorem]{Remark}

\newcommand{\rfb}[1]{\mbox{\rm
		(\ref{#1})}\ifx\undefined\stillediting\else:\fbox{$#1$}\fi}

\newfont{\roma}{cmr10 scaled 1200}

\newcommand{\nline}  {{\mathbb N}}
\newcommand{\rline}  {{\mathbb R}}

\newcommand{\tline}  {{\mathbb T}}

\newcommand{\dd}   {{\rm d}\hbox{\hskip 0.5pt}}
\newcommand{\Ascr} {{\mathcal A}}
\newcommand{\Bscr} {{\mathcal B}}

\newcommand{\Dscr} {{\mathcal D}}

\newcommand{\Lscr} {{\mathcal L}}

\newcommand{\Wscr} {{\mathcal W}}

\newcommand{\mm}    {{\hbox{\hskip 0.5pt}}}
\newcommand{\m}     {{\hbox{\hskip 1pt}}}

\newcommand{\bluff} {{\hbox{\raise 15pt \hbox{\mm}}}}
\newcommand{\sbluff}{{\hbox{\raise  7pt \hbox{\mm}}}}

\newcommand{\FORALL} {{\hbox{$\hskip 11mm \forall \;$}}}

\title{\LARGE \bf Flatness-based motion planning for a non-uniform moving cantilever Euler-Bernoulli beam with a tip-mass
}

\author{Soham Chatterjee, Aman Batra and Vivek Natarajan
	\thanks{A. Batra is supported by the Prime Minister's Research Fellowship via the grant RSPMRF0262.}
	\thanks{S. Chatterjee ({\tt\small soham.chatterjee@iitb.ac.in}), A. Batra ({\tt\small aman.batra@iitb.ac.in}) and V. Natarajan ({\tt\small vivek.natarajan@iitb.ac.in}) are with the Centre for Systems and Control Engineering, Indian Institute of Technology Bombay, Mumbai 400076, India.}
}

\begin{document}
\maketitle
\thispagestyle{empty}
\pagestyle{empty}
\begin{abstract}
Consider a non-uniform Euler-Bernoulli beam with a tip-mass at one end and a cantilever joint at the other end. The cantilever joint is not fixed and can itself be moved along an axis perpendicular to the beam. The position of the cantilever joint is the control input to the beam.
The dynamics of the beam is governed by a coupled PDE-ODE model with boundary input. On a natural state-space, there exists a unique state trajectory for this beam model for every initial state and each smooth control input which is compatible with the initial state. In this paper, we study the motion planning problem of transferring the beam from an initial state to a final state over a prescribed time-interval. We address this problem by extending the generating functions approach to flatness-based control, originally proposed in the literature for motion planning of parabolic PDEs, to the beam model. We prove that such a transfer is possible if the initial and final states belong to a cer-tain set, which also contains steady-states of the beam. 
We illust-rate our theoretical results using simulations and \vspace{2mm} experiments.

\end{abstract}

\section{Introduction} \label{intro}

Consider the following model of a non-uniform moving cantilever Euler-Bernoulli beam with a tip-mass at one end and a cantilever joint at the other end:
\begin{align}
 & \rho(x) w_{tt}(x,t)+(EI w_{xx})_{xx}(x,t) = 0, \label{eq:beam1}\\
 & \hspace{2mm} m w_{tt}(0,t)+(EI w_{xx})_x(0,t)=0, \label{eq:beam2}\\
 & \hspace{2mm} J w_{xtt}(0,t) - EI(0)w_{xx}(0,t) = 0, \label{eq:beam3} \\
 & \hspace{2mm} w(L,t) = f(t), \qquad w_x(L,t)=0. \label{eq:beam4}
\end{align}
Here $L>0$ is the length of the beam, $w(x,t)$ is the displacement of the beam at the location $x\in [0,L]$ and time $t\in(0,\infty)$, $\rho(x)$ and $EI(x)$ are the mass per unit length and flexural rigidity, respectively, of the beam at $x\in[0,L]$, and $m>0$ and $J>0$ are the mass and moment of inertia, respectively, of the tip-mass located at the $x=0$ end of the beam. We suppose that $\rho \in C^4[0,L]$, $EI \in C^4[0,L]$ and they are strictly positive, i.e $\inf_{x\in[0,L]}\rho(x)>0$ and $\inf_{x\in[0,L]}EI(x)>0$. The displacement of the cantilever joint at the $x=L$ end of the beam is determined by the scalar control input $f$. The coupled PDE-ODE model \eqref{eq:beam1}-\eqref{eq:beam4} governs the dynamics of engineering systems which have a moving cantilever beam such as single-axis flexible cartesian robots. Figure 1 shows an experimental setup consisting of a non-uniform moving cantilever beam with a tip-mass. \vspace{-3mm}

$$\includegraphics[width=0.45\textwidth]{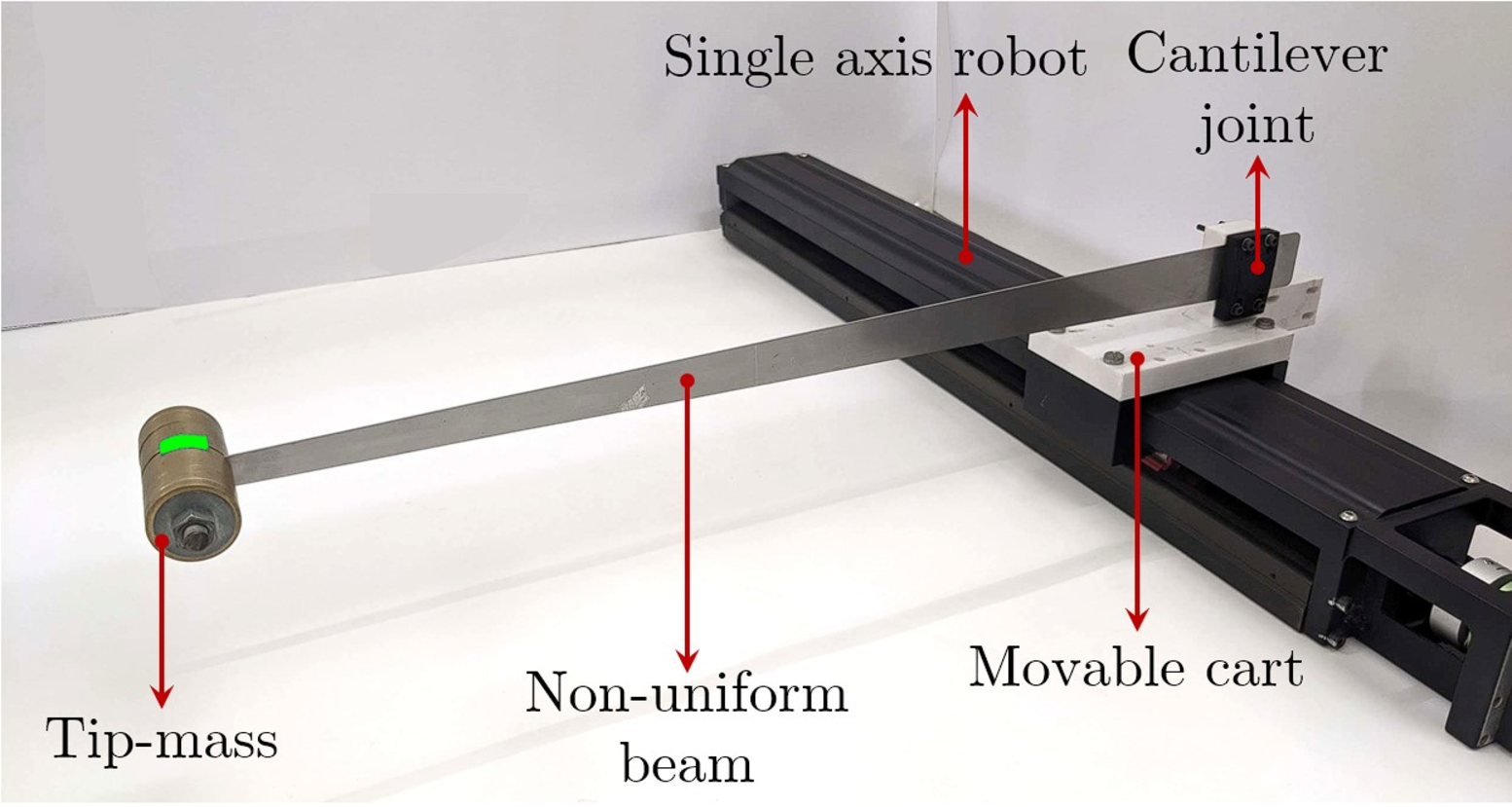} $$
\centerline{ \parbox{3.2in}{\small
Figure 1. Experimental setup of a non-uniform moving cantilever beam with tip-mass. The green marker on the tip-mass is used to track the tip-mass position with a high speed \vspace{3mm} camera. }}

There exists a natural state-space $Z$ for the non-uniform moving cantilever Euler-Bernoulli beam model \eqref{eq:beam1}-\eqref{eq:beam4} in which it has a unique state trajectory for every initial state $z_0$ and each
smooth control input $f$ which is compatible with $z_0$, see Section \ref{sec2}. In this paper, we study the motion planning problem of computing a control input $f$ which transfers  \eqref{eq:beam1}-\eqref{eq:beam4} from an initial state $z_0$ to a final state $z_T$ over a prescribed time-interval $[0,T]$. We solve this problem for a set of initial and final states in $Z$, which includes the steady-states of the beam.

Over the past three decades, flatness has emerged as a predominant technique for solving motion planning problems for dynamical systems governed by PDEs, including flexible structures. The main idea of flatness is to express the state and the input of the PDE as an infinite linear combination of a flat output and its time derivatives. Then, based on the desired motion, an appropriate trajectory is selected for the flat output using which the input necessary to execute the motion is computed. Using the transform approach (Laplace transform or Mikusi\'nski's operational calculus) to flatness, motion planning problems have been solved for Euler-Bernoulli beams with fixed cantilever joints in \cite{MeThKu:2008}, rotating cantilever joints in \cite{AoFlMoRoRu:1997}, \cite{BaLy:2008}, \cite{LyWa:2004} and \cite{RiTz:2009} and translating cantilever joints in \cite{BaUlRu:2011}. More recently, the power series approach to flatness has been used in \cite{BaZh:2014}, \cite{BaZhZh:2016}, \cite{KoGeSc:2022} to solve motion planning problems for Euler-Bernoulli beams with some other boundary conditions. 
The Riesz spectral approach to flatness, see \cite{Me:2012}, is in general applicable to PDEs with spatially-varying coefficients. However, addressing motion planning problems for non-uniform Euler-Bernoulli beams using this approach requires spectral assumptions that are hard to verify \cite{MeScKu:2010}. Moreover, the admissibility assumption on the control operator in \cite{MeScKu:2010} does not hold for our beam model \!\! \eqref{eq:beam1}-\eqref{eq:beam4}.

Motion planning of 1D PDEs using the generating functions approach to flatness involves solving a sequence of initial value ODEs recursively to obtain the generating functions and then expressing the input and solution of the PDE in terms of these functions. In principle, this approach is well-suited for PDEs with spatially-varying coefficients and it has been used to address a motion planning problem (null control problem) for 1D parabolic PDEs with highly irregular coefficients in \cite{MaRoRo:2016}. In the present work, we extend this approach to solve the motion planning problem described above for the non-uniform Euler-Bernoulli beam model \eqref{eq:beam1}-\eqref{eq:beam4}. 
The extension inherently leads to a certain infinite-order differential equation which we solve by establishing that a pair of infinite-order differential operators commute. We show that if the initial and final states belong to a certain subspace of the state-space which contains steady-states of the beam, then our motion planning problem has a solution for any prescribed time-interval. We illustrate our theoretical results using both simulations and experiments.

The rest of the paper is organized as follows: In Section \ref{sec2} we  establish the well-posedness of the Euler-Bernoulli beam model \eqref{eq:beam1}-\eqref{eq:beam4}. We define the generating functions for the beam model in Section \ref{sec3} and derive some estimates for them. Section \ref{sec4} contains our solution to the motion planning problem and in Section \ref{sec5} we present our numerical and experimental results.

\section{Well-posedness} \label{sec2}

In this section we establish the existence and uniqueness of solutions to the coupled PDE-ODE model \eqref{eq:beam1}-\eqref{eq:beam4}. Let $H^n(0,L)$ denote the usual Sobolev space of order $n$ on the interval $(0,L)$. A natural choice of state space for \eqref{eq:beam1}-\eqref{eq:beam4} is the Hilbert space
$$ Z = \{[u \ \  v \ \ \alpha \ \ \beta] \in H^2(0,L)\times L^2(0,L)\times \rline\times \rline \big| u_x(L)=0 \}.$$
For $z_1=[u_1 \ v_1 \ \alpha_1 \ \beta_1]$ and $z_2=[u_2 \ v_2 \ \alpha_2 \ \beta_2]$ in $Z$, the inner product $\langle z_1, z_2 \rangle_Z = \int_0^L EI(x) u_{1,xx}(x) u_{2,xx}(x) \dd x + \int_0^L u_1(x) u_2(x) \dd x  + \int_0^L \rho(x) v_1(x) v_2(x) \dd x + m\alpha_1\alpha_2 + J \beta_1\beta_2$.  We now define the notion of classical solutions for \eqref{eq:beam1}-\eqref{eq:beam4} in $Z$. \vspace{1mm}

\begin{definition}\label{def:classical}
\!Given $T\!>\!0$, $z_0\! \in\! Z$ and $f\!\in\! C([0,T];\rline)$, a function $z\in C([0,T];Z)$ is a \emph{classical solution} of \eqref{eq:beam1}-\eqref{eq:beam4} on the time interval $[0,T]$ for the initial state $z_0$ and input $f$ if $z(0)=z_0$ and
\begin{equation} \label{eq:wdetz}
 z(t)=[w(\cdot,t) \ \ w_t(\cdot,t) \ \ w_t(0,t) \ \ w_{xt}(0,t)] \qquad \forall\m t\in[0,T],
\end{equation}
where $w\in C([0,T];H^4(0,L))\cap C^1([0,T];H^2(0,L))\cap C^2([0,T];L^2(0,L))$ with $w(0,\cdot),w_x(0,\cdot)\in C^2([0,T];\rline)$ and $w$ and $f$ satisfy \eqref{eq:beam1}-\eqref{eq:beam4} for each $t\in(0,T)$. \vspace{1mm} \hfill$\square$
\end{definition}

Consider the following dense subspace of $Z$:
\begin{align}
 V =& \{[u \ \ v \ \ \alpha \ \ \beta] \in H^4(0,L)\times H^2(0,L) \times \rline \times \rline \m\big|\m \nonumber\\
 &\quad  u_x(L)=v_x(L)=0,\, v(0)=\alpha, \, v_x(0)=\beta\}. \label{eq:spaceV}
\end{align}
In the next proposition we will establish the existence and uniqueness of a classical solution for \eqref{eq:beam1}-\eqref{eq:beam4} when the initial state belongs to $V$ and the input $f$ is smooth and compatible with the initial state. Since the proof of the proposition is based on standard techniques, we omit detailed explanations in the proof.

\begin{proposition}\label{pr:wellposed}
Let $T>0$, an input $f\in C^3([0,T];\rline)$ and an initial state $z_0 = [u_0 \ \ v_0 \ \ \alpha_0 \ \ \beta_0]\in V$ with $f(0)=u_0(L)$ and $\dot f(0)=v_0(L)$ be given. There exists a unique classical solution $z\in C([0,T];Z)$ of \eqref{eq:beam1}-\eqref{eq:beam4} on the time interval $[0,T]$ for the initial state $z_0$ and input $f$.
\end{proposition}

\begin{proof}
Choose $\nu\in C^\infty[0,L]$ such that $\nu(0)=\nu_x(0)=\nu_{xx}(0)= \nu_{xxx}(0)=\nu_x(L)=0$ and $\nu(L)=1$. Replacing $w(x,t)$ with $\tilde w(x,t)+\nu(x)f(t)$ formally in \eqref{eq:beam1}-\eqref{eq:beam4} we get that for $x\in (0,L)$ and $t>0$,
\begin{align}
  &\rho(x)\tilde w_{tt}(x,t)+ (EI \tilde w_{xx})_{xx}(x,t)\nonumber\\
  &\hspace{2mm} + (EI(\nu_{xx})_{xx}(x)f(t)+\rho(x)\nu(x)\ddot{f}(t)=0, \label{eq:wbeam1}\\
  &m \tilde w_{tt}(0,t)+(EI \tilde w_{xx})_x(0,t)=0, \label{eq:wbeam2} \\
  &J\tilde w_{xtt}(0,t) - EI(0)\tilde w_{xx}(0,t) = 0, \label{eq:wbeam3} \\
  &\hspace{5mm}\tilde w(L,t) =0, \quad  \tilde w_x(L,t)=0. \label{eq:wbeam4}
\end{align}
The subspace $\tilde Z = \{[u \ v \ \alpha \ \beta] \in Z \m\big|\m u(L)=0\}$ is a closed subset of $Z$ and is itself a Hilbert space. The coupled PDE-ODE system \eqref{eq:wbeam1}-\eqref{eq:wbeam4} can be written as an abstract evolution equation on $\tilde Z$ as follows:
\begin{equation}\label{eq:evol}
 \dot{\tilde{z}}(t) = \Ascr \tilde z(t) + \Bscr\bigl[ f(t) \ \ \ddot f(t)\bigr] \qquad \forall \, t>0.
\end{equation}
The operators $\Ascr$ and $\Bscr$ are defined as follows: The domain of $\Ascr$ is $\Dscr(\Ascr) = \{[u \ v \ \alpha \ \beta] \in V \m\big|\m u(L)=v(L)=0\}$ and $\Ascr [u \ v \ \alpha \ \beta] = [v \ \ -\frac{(EI u_{xx})_{xx}}{\rho} \ \ -\frac{(EI u_{xx})_x(0)}{m} \ \ \frac{EI(0)u_{xx}(0)}{J}]$ for $[u \ v \ \alpha \ \beta]\in \Dscr(\Ascr)$  and $\Bscr: \rline\times \rline\mapsto \tilde Z$ is a bounded linear operator with $\Bscr [ a \ b] = [0 \ \frac{-a(EI  \nu_{xx})_{xx} -b\rho\nu}{\rho} \ 0 \ 0]$. The operator $\Ascr$ generates a $C_0$-semigroup $\tline$ on $\tilde Z$, see \cite[Section 5]{ZhWe:2011}.

Recall $z_0$ and $f$ and their properties from the statement of the proposition. Define $\tilde z_0=z_0-[\nu f(0) \ \nu\dot f(0) \ 0 \ 0]$. Then $\tilde z_0\in\Dscr(\Ascr)$. From \cite[Chapter 4, Section 2]{Pazy:1983} it follows that there exists a unique function $\tilde z \in C^1([0,T];\tilde Z)\cap C([0,T];\Dscr(\Ascr))$, given by the expression
\begin{equation} \label{eq:tildems}
 \tilde z(t) = \tline_t \tilde z_0 + \int_0^t \tline_{t-\tau}\Bscr \bigl[f(\tau) \ \ \ddot f(\tau)\bigr]\dd \tau,
\end{equation}
such that $\tilde z(0)=\tilde z_0$ and $\tilde z$ satisfies \eqref{eq:evol} for each $t\in (0,T)$. Equivalently, there exists a unique function $\tilde w \in C([0,T];H^4(0,L))\cap C^1([0,T];H^2(0,L))\cap C^2([0,T];L^2(0,L))$ with $\tilde w(0,\cdot), \tilde w_x(0,\cdot)\in C^2([0,T];\rline)$ such that
$ \tilde z(t) = [\tilde w(\cdot,t) \ \ \tilde w_t(\cdot,t) \ \ \tilde w_t(0,t) \ \ \tilde w_{xt}(0,t)] $ satisfies $\tilde z(0)=\tilde z_0$ and $\tilde w$ and $f$ satisfy \eqref{eq:wbeam1}-\eqref{eq:wbeam4} for each $t\in(0,T)$. Defining $w=\tilde w + \nu f$, it follows that there exists a unique function $w \in C([0,T];H^4(0,L))\cap C^1([0,T];H^2(0,L))\cap C^2([0,T];L^2(0,L))$ with $w(0,\cdot), w_x(0,\cdot)\in C^2([0,T];\rline)$ such that $z(t) = [w(\cdot,t) \ \ w_t(\cdot,t) \ \ w_t(0,t) \ \ w_{xt}(0,t)]$ satisfies $z(0)=z_0$ and $w$ and $f$ satisfy \eqref{eq:beam1}-\eqref{eq:beam4} for each $t\in(0,T)$. This completes the proof of the proposition.
\end{proof}

It is easy to see using \eqref{eq:tildems} that the unique classical solution $z$ in the above proof is given by the following formula: $z(t) = \tline_t \tilde z_0 + \int_0^t \tline_{t-\tau}\Bscr \bigl[f(\tau) \ \ \ddot f(\tau)\bigr]\dd \tau +[\nu f(t) \ \ \nu\dot f(t) \ \ 0 \ \ 0].$
The expression on the right side of this formula makes sense even for inputs $f\in C^2([0,T];\rline)$ and initial states $z_0 = [u_0 \ v_0 \ \alpha_0 \ \beta_0]\in Z$ satisfying $f(0)=u_0(L)$. Indeed it defines the unique strong solution of \eqref{eq:beam1}-\eqref{eq:beam4} for such inputs and initial states. But since the initial states and inputs of \eqref{eq:beam1}-\eqref{eq:beam4} considered in the rest of the paper satisfy the assumptions in Proposition \ref{pr:wellposed}, we will not discuss strong solutions any further.



\section{Generating functions} \label{sec3}

For each $s>0$, the Gevrey class $G_s[0,T]$ is the space of all the functions $y\in C^\infty[0,T]$ which satisfy the estimate $\sup_{t\in [0,T]} |y^{(m)}(t)|\leq D^{m+1} (m!)^s$ for some $D>0$ and all integers $m\geq 0$. Here $y^{(m)}$ denotes the $m^{\rm th}$-derivative of $y$.

In this paper we solve a motion planning problem for the beam model \eqref{eq:beam1}-\eqref{eq:beam4} by building on the generating functions approach to flatness proposed for 1D parabolic PDEs in \cite{MaRoRo:2016}. Accordingly we suppose that the solution $w$ of \eqref{eq:beam1}-\eqref{eq:beam4} on the interval $[0,T]$ can be expressed as
\begin{equation}\label{eq:formalsoln}
 w(x,t) = \sum_{k\geq 0} g_k(x)y_1^{(2k)}(t) + \sum_{k\geq 0} h_k(x)y_2^{(2k)}(t)
\end{equation}
for all $x\in[0,L]$ and $t\in [0,T]$. Here $g_k$ and $h_k$ belonging to $H^4(0,L)$ are the generating functions (see Proposition \ref{pr:genfnest}) and $y_1, y_2\in G_s[0,T]$ with $1<s<2$ are the flat outputs. We remark that while the solution of 1D parabolic PDEs can be expressed using a single flat output $y$, see \cite{MaRoRo:2016}, we need two flat outputs $y_1, y_2$ to express the solution of the beam model \eqref{eq:beam1}-\eqref{eq:beam4}. These flat outputs cannot be chosen independently and must satisfy \eqref{eq:2ndinput}.

The generating functions $g_k$ are obtained by solving a sequence of fourth-order linear ODEs recursively, on the interval $x\in[0,L]$, as follows:
\begin{equation}\label{eq:g0}
 g_0(x) = 1,
\end{equation}
$g_1$ is obtained by solving the ODE
\begin{align}
  &\quad(EI g_{1,xx})_{xx}(x) + \rho(x)g_0(x)=0, \label{eq:g1ODE}\\
  &\qquad g_1(0) = 0, \qquad  g_{1,x}(0) = 0, \label{eq:g1IC1}\\
  &g_{1,xx}(0) = 0, \qquad (EI g_{1,xx})_x(0) = -m,\label{eq:g1IC2}
\end{align}
and $g_k$ for $k\geq 2$ is obtained by solving the ODE
\begin{align}
  &\quad(EI g_{k,xx})_{xx}(x) + \rho(x)g_{k-1}(x)=0, \label{eq:gkODE}\\
  &\qquad g_k(0) = 0, \qquad  g_{k,x}(0) = 0,\label{eq:gkIC1}\\
  &\ \ g_{k,xx}(0) = 0, \qquad (EI g_{k,xx})_x(0) = 0.\label{eq:gkIC2}
\end{align}
The generating functions $h_k$ are also obtained by solving another sequence of fourth-order linear ODEs recursively, on the interval $x\in[0,L]$, as follows:
\begin{equation}\label{eq:h0}
 h_0(x) = x,
\end{equation}
$h_1$ is obtained by solving the ODE
\begin{align}
  &\ \ \quad (EI h_{1,xx})_{xx}(x) + \rho(x) h_{0}(x)=0,  \label{eq:h1ODE}\\
  &\ \ \qquad h_1(0) = 0, \qquad  h_{1,x}(0) = 0, \label{eq:h1IC1}\\
  &EI(0)h_{1,xx}(0) = J, \qquad (EI h_{1,xx})_x(0) = 0, \label{eq:h1IC2}
\end{align}
and $h_k$ for $k\geq 2$ is obtained by solving the ODE
\begin{align}
  &\quad (EI h_{k,xx})_{xx}(x)\! + \rho(x)h_{k-1}(x)=0, \label{eq:hkODE}\\
  &\qquad h_k(0) = 0, \qquad  h_{k,x}(0) = 0, \label{eq:hkIC1}\\
  & h_{k,xx}(0) = 0, \qquad (EI h_{k,xx})_x(0) = 0.\label{eq:hkIC2}
\end{align}

In the following proposition we show that the generating functions $g_k$ and $h_k$ belong to $H^4(0,L)$ and derive some estimates for them. Using these estimates, in Proposition \ref{pr:wisclasic} we show that if $y_1$ and $y_2$ satisfy \eqref{eq:2ndinput}, then the function $z\in C([0,T];Z)$ determined by the function $w$ in \eqref{eq:formalsoln} via the expression \eqref{eq:wdetz} is the classical solution of \eqref{eq:beam1}-\eqref{eq:beam4} for the initial state $z_0 = z(0)$ and input $f(t)=w(L,t)$.

\begin{proposition}\label{pr:genfnest}
For each $k\geq1$ the generating functions $g_k$ and $h_k$ belong to $H^4(0,L)$ and there exist positive constants $R_1$ and $R_2$ independent of $k$ such that the following estimates hold for $x\in [0,L]$:
\begin{align}
 &|g_k(x)|\leq \frac{R_1^k\, x^{4k-1}}{(4k-1)!}, \qquad   |g_{k,x}(x)| \leq \frac{R_1^k \,x^{4k-2}}{(4k-2)!}, \label{eq:estggx}\\[0.5ex]
 &|h_k(x)|\leq \frac{R_2^k\, x^{4k-2}}{(4k-2)!}, \qquad |h_{k,x}(x)| \leq \frac{R_2^k\, x^{4k-3}}{(4k-3)!}. \label{eq:esthhx}
\end{align}
\end{proposition}
\vspace{1mm}

\begin{proof}
\!\!Recall $g_0\!=\!1$ from \eqref{eq:g0}. Solving \eqref{eq:g1ODE}-\eqref{eq:g1IC2} we get\!\!
\begin{align}
  g_1(x) =& -\int_0^x \!\int_0^{s_1}\!\int_0^{s_2}\! \int_0^{s_3 }\! \frac{\rho(s_4)}{EI(s_2)}\dd s_4 \dd s_3 \dd s_2 \dd s_1 \nonumber\\
  &\hspace{3mm}-m\int_0^x \int_0^{s_1} \int_0^{s_2}\frac{1}{EI(s_2)}\dd s_3 \dd s_2 \dd s_1 . \label{eq:g1sol}
\end{align}
Since $EI$ and $\rho$ are strictly positive functions in $C^4[0,L]$ it follows from the above equation that $g_1 \in H^4(0,L)$ and the estimates in \eqref{eq:estggx} hold for $k=1$ with
\begin{equation}\label{eq:R1defn}
  R_1 = \Bigg(\frac{\max_{x\in [0,L]}\rho(x)+m}{\min_{x\in [0,L]} EI(x)} \Bigg) \max\{1,L\}.
\end{equation}
Solving \eqref{eq:gkODE}-\eqref{eq:gkIC2} for $g_k$ we get that for $k\geq2$,
\begin{equation}\label{eq:gksol}
 g_k(x) = -\int_{0}^{x}\!\int_{0}^{s_1}\!\!\int_{0}^{s_2}\!\! \int_{0}^{s_3}\! \frac{\rho(s_4)g_{k-1}(s_4)}{EI(s_2)}\dd s_4 \dd s_3 \dd s_2 \dd s_1.
\end{equation}
Let $R_1$ be as defined in \eqref{eq:R1defn}. Suppose that $g_k\in H^4(0,L)$ and the estimates in \eqref{eq:estggx} hold for some $k=n$ with $n\geq 1$. We can then conclude that the same is true for $k=n+1$ using \eqref{eq:gksol} with $k=n+1$, the estimates in \eqref{eq:estggx} with $k=n$ and the fact that $EI, \rho \in C^4[0,L]$ are strictly positive. So from the principle of mathematical induction it follows that $g_k\in H^4(0,L)$ and the estimates in \eqref{eq:estggx} hold for all $k\geq1$.


Next we will prove the estimates in \eqref{eq:esthhx}. Recall from \eqref{eq:h0} that $h_0(x)=x$. Solving \eqref{eq:h1ODE}-\eqref{eq:h1IC2} we get
\begin{align}
 h_1(x) =& -\int_0^x \!\int_0^{s_1}\!\int_0^{s_2}\! \int_0^{s_3 }\! \frac{s_4\rho(s_4)}{EI(s_2)}\dd s_4 \dd s_3 \dd s_2 \dd s_1 \nonumber\\
 &\hspace{5mm}+J\int_0^x \int_0^{s_1}\frac{1}{EI(s_2)}\dd s_2 \dd s_1. \label{eq:h1sol}
\end{align}
Since $EI$ and $\rho$ are strictly positive functions in $C^4[0,L]$ it follows from the above equation that $h_1 \in H^4(0,L)$ and the estimates in \eqref{eq:esthhx} hold for $k=1$ with \vspace{-1mm}
\begin{equation}\label{eq:R2defn}
  R_2 = \Bigg(\frac{\max_{x\in [0,L]}\rho(x)+J}{\min_{x\in [0,L]} EI(x)} \Bigg) \max\{1,L^3\}. \vspace{-1mm}
\end{equation}
Solving \eqref{eq:hkODE}-\eqref{eq:hkIC2} for $h_k$ we get that for $k\geq2$, \vspace{-0.5mm}
$$ h_k(x) = -\int_{0}^{x}\!\int_{0}^{s_1}\!\!\int_{0}^{s_2}\!\! \int_{0}^{s_3} \!\!\frac{\rho(s_4)h_{k-1}(s_4)}{EI(s_2)}\dd s_4 \dd s_3 \dd s_2 \dd s_1. \vspace{-0.5mm}$$
The claim that $h_k\in H^4(0,L)$ and the estimates in \eqref{eq:esthhx} hold for all $k\geq1$ with $R_2$ given in \eqref{eq:R2defn} can be established by mimicking the induction argument given below \eqref{eq:gksol}. \vspace{1mm}
\end{proof}

\begin{proposition}\label{pr:wisclasic}
Fix $T>0$ and $s\in (1,2)$. Suppose that the functions $y_1, y_2\in G_s[0,T]$ satisfy the following infinite order differential equation: for $t\in[0,T]$, \vspace{-0.5mm}
\begin{equation}\label{eq:2ndinput}
 \sum_{k\geq 0} g_{k,x}(L)y_1^{(2k)}(t) + \sum_{k\geq 0} h_{k,x}(L)y_2^{(2k)}(t) = 0. \vspace{-1mm}
\end{equation}
Then $w$ given in \eqref{eq:formalsoln} belongs to $C^\infty([0,T];C^4[0,L])$ so that $w(0,\cdot)$, $w(L,\cdot)$  and $w_x(0,\cdot)$ belong to $C^\infty([0,T];\rline)$. Furthermore, the function $z\in C([0,T];Z)$ determined by $w$ via \eqref{eq:wdetz} is the unique classical solution of \eqref{eq:beam1}-\eqref{eq:beam4} for the initial state $z_0=z(0)$ and input $f(t)=w(L,t)$.
\end{proposition}

\begin{proof}
Differentiating the expression for $g_k$ in \eqref{eq:gksol} as required and then using the estimate for $g_{k-1}$ from \eqref{eq:estggx} it follows that the functions $g_k$, $g_{k,x}$, $g_{k,xx}$, $g_{k,xxx}$ and $g_{k,xxxx}$ are uniformly bounded on $[0,L]$ by $C_g^k/(4k-5)!$. Here $C_g>0$ is independent of $k$. Similar estimates can be obtained for  $h_k$, $h_{k,x}$, $h_{k,xx}$, $h_{k,xxx}$ and $h_{k,xxxx}$.  Using these estimates and the fact that $y_1,y_2$ are in $G_s[0,T]$ for some $s\in (1,2)$ (which implies that $y_1^{(k)}$ and $y_2^{(k)}$ are uniformly bounded on $[0,T]$ by $C_y^{k+1}(k!)^s$ for some $C_y>0$ and all $k\geq0$), it follows via the Weierstrass M-test and the ratio test that the series for $w$ and its derivatives with respect to $x$ and $t$ (obtained by termwise differentiation of the series in \eqref{eq:formalsoln}) converge uniformly on $[0,L]\times[0,T]$. This implies that the derivatives of $w$ with respect to $x$ (up to four times) and $t$ (any number of times) are nothing but the series obtained by termwise differentiation of the series in \eqref{eq:formalsoln} and $w$ has the regularity mentioned in the statement of this proposition.

Using \eqref{eq:g0}, \eqref{eq:g1ODE}, \eqref{eq:gkODE} and \eqref{eq:h0}, \eqref{eq:h1ODE}, \eqref{eq:hkODE} it is easy to check that the series corresponding to $w_{tt}$ and $-(EI w_{xx})_{xx}$ are the same and hence $w$ satisfies the PDE \eqref{eq:beam1}. Taking $x=0$ in the series for $w_{tt}$, $(EI w_{xx})_x$, $w_{xtt}$ and $EIw_{xx}$ and using the initial conditions for $g_k$ in \eqref{eq:g1IC1}, \eqref{eq:g1IC2}, \eqref{eq:gkIC1}, \eqref{eq:gkIC2} and the initial conditions for by $h_k$ in \eqref{eq:h1IC1}, \eqref{eq:h1IC2}, \eqref{eq:hkIC1}, \eqref{eq:hkIC2} it follows that $w$ satisfies the boundary conditions \eqref{eq:beam2}-\eqref{eq:beam3}. Finally taking  $f(t)=w(L,t)$ and using \eqref{eq:2ndinput} (which means $w_x(L,t)=0$) we get that $w$ satisfies the boundary condition \eqref{eq:beam4}. In summary, $w$ satisfies \eqref{eq:beam1}-\eqref{eq:beam4} and has the desired regularity so that $z$ given by \eqref{eq:wdetz} is a classical solution of \eqref{eq:beam1}-\eqref{eq:beam4} for the initial state $z_0=z(0)$ and input $f(t)$. Recall the set $V$ from \eqref{eq:spaceV}. Differentiating \eqref{eq:2ndinput} with respect to $t$ we get $w_{xt}(L,t)=0$. Using this, \eqref{eq:2ndinput} and the regularity of $w$ we can conclude that $z_0\in V$. Also $f$ is compatible with $z_0$ by definition. The uniqueness of the classical solution $z$ now follows from Proposition \ref{pr:wellposed}.
\end{proof}


\section{Motion planning} \label{sec4} 

We present our main results on the motion planning problem for the beam model \eqref{eq:beam1}-\eqref{eq:beam4} in this section,
see Theorem \ref{th:main_result} and Remark \ref{rm:ss}. In the following proposition, we first describe an approach for constructing functions $y_1$ and $y_2$ that satisfy \eqref{eq:2ndinput}. Below we will need the estimate
\begin{equation} \label{eq:binom}
 (a+b)! \leq 2^{a+b} a! b! \FORALL a,b\in\nline.
\end{equation}
This estimate follows from the fact that the $(a+1)^{\rm th}$-term in the binomial expansion of $(1+1)^{a+b}$ is  less than $(1+1)^{a+b}$.

\begin{proposition}\label{pr:psolution}
Fix $s\in(1,2)$. Let the operators $\Lscr_1$ and $\Lscr_2$ be defined as follows: For $p \in G_s[0,T]$,
$$\Lscr_1 p =\sum_{k\geq 0} g_{k,x}(L) p^{(2k)}, \qquad \Lscr_2 p=\sum_{k\geq 0} h_{k,x}(L) p^{(2k)}.$$
Then $y_1=\Lscr_2 p$ and $y_2=-\Lscr_1 p$ belong to $G_s[0,T]$ for each $p \in G_s[0,T]$ and satisfy \eqref{eq:2ndinput}.
\end{proposition}
\vspace{1mm}
\begin{proof}
Let $p\in G_s[0,T]$. Then
\begin{equation}\label{eq:pest}
 \sup_{t\in [0,T]}|p^{(k)}(t)|\leq D^{k+1} (k\m!)^s \FORALL k\geq 0
\end{equation}
and some constant $D>0$. Using this estimate and the estimate for $h_{k,x}$ in \eqref{eq:esthhx} we can conclude by applying the Weierstrass M-test and the ratio test that for each $n\geq 0$ the series for $y_1^{(n)}$, obtained by differentiating the series for $\Lscr_2 p$ termwise $n$-times with respect to $t$, converges uniformly on $[0,T]$. Hence $y_1^{(n)} = \Lscr_2 p^{(n)}$ and using \eqref{eq:pest} and \eqref{eq:esthhx} we get
$$ \sup_{t\in [0,T]}|y_1^{(n)}(t)|\leq \sum_{k\geq 0} R_2^k L^{4k-3} D^{2k+n+1}\frac{((2k+n)!)^s}{(4k-3)!}. $$
Using \eqref{eq:binom} with $a=2k$ and $b=n$ to bound $(2k+n)!$ in the above inequality we get $\sup_{t\in [0,T]}|y_1^{(n)}(t)|\leq 2^{ns}D^n (n!)^s \sum_{k\geq 0} C_0^{k+1} ((2k)!)^s\big/(4k-3)!$ for some $C_0>0$. Since $s<2$, applying the ratio test it follows that the series $\sum_{k\geq 0} C_0^{k+1} ((2k)!)^s\big/(4k-3)!$ converges and therefore $\sup_{t\in [0,T]}|y_1^{(n)}(t)|\leq C^{n+1} (n!)^s$ for a $C>0$ and all $n\geq0$, i.e. $y_1\in G_s[0,T]$. We can similarly show that \vspace{1mm} $y_2\in G_s[0,T]$.

Note that $\Lscr_1 \Lscr_2 p = \sum_{k\geq 0}\sum_{j\geq 0} g_{k,x}(L) h_{j,x}(L)p^{(2k+2j)}$ and $\Lscr_2 \Lscr_1 p = \sum_{j\geq 0}\sum_{k\geq 0} g_{k,x}(L) h_{j,x}(L)p^{(2k+2j)}$ for each $p\in G_s[0,T]$. So $\Lscr_1 \Lscr_2 p $ and $\Lscr_2 \Lscr_1 p$ are both double sums which only differ in the order of summation. We claim that the double sum $\sum_{l \geq 0} \sum_{j+k=l} g_{j,x}(L)h_{k,x}(L) p^{(2l)}(t)$ is absolutely convergent for each $t\in[0,T]$. Indeed, using \eqref{eq:pest}, \eqref{eq:estggx} and \eqref{eq:esthhx} we get
\begin{align*}
 &\sum_{l \geq 0} \sum_{j+k=l}| g_{j,x}(L)h_{k,x}(L) p^{(2l)}(t)| \\
 \leq & \sum_{l \geq 0} \sum_{j+k=l} \frac{R_1^j R_2^k L^{4j-2} L^{4k-3} D^{2l+1} ((2l)!)^s}{(4j-2)!(4k-3)!}\\
 \leq & \sum_{l\geq 0}\sum_{j+k=l}\frac{R^{l+1} ((2l)!)^s}{(4l-5)!} = \sum_{l\geq 0}\frac{R^{l+1} ((2l)!)^s (l+1)}{(4l-5)!} <\infty.
\end{align*}
Here $R>0$ is some constant, the second inequality is derived by using the estimate in \eqref{eq:binom} with $a=4j-2$ and $b=4k-3$ and the last inequality is obtained by applying the ratio test. Therefore from Fubini's theorem it follows that the double sum $\sum_{j,k\geq 0} g_{j,x}(L)h_{k,x}(L) p^{(2l)}$ is independent of the order of summation, i.e. $\Lscr_1 \Lscr_2 p=\Lscr_2 \Lscr_1 p$. So $y_1=\Lscr_2 p$ and $y_2=-\Lscr_1 p$ satisfy $\Lscr_1 y_1 - \Lscr_2 y_2=0$ or equivalently \eqref{eq:2ndinput}. This completes the proof of the proposition.
\end{proof}

\begin{remark} \label{rm:summary}
We can rewrite \eqref{eq:formalsoln} concisely as
\begin{equation} \label{eq:formalsolnRE}
 w(x,t) = [\Wscr_1 y_1](x,t) + [\Wscr_2 y_2](x,t)
\end{equation}
for $x\in[0,L]$ and $t\in[0,T]$. Here the operators $\Wscr_1$ and $\Wscr_2$ are defined as follows: $[\Wscr_1 y](x,t) = \sum_{k\geq 0} g_k(x) y^{(2k)}(t)$ and $[\Wscr_2 y](x,t) = \sum_{k\geq 0} h_k(x) y^{(2k)}(t)$ for $y\in G_s[0,T]$ with $s\in(1,2)$. For any $p\in G_s[0,T]$ with $s\in(1,2)$, it follows from Proposition \ref{pr:psolution} that $y_1 = \Lscr_2 p$ and $y_2 = -\Lscr_1 p$ are in $G_s[0,T]$ and satisfy \eqref{eq:2ndinput}. Letting $y_1 = \Lscr_2 p$ and $y_2 = -\Lscr_1 p$ in \eqref{eq:formalsolnRE} and appealing to Proposition \ref{pr:wisclasic} we can conclude that $w = \Wscr_1 \Lscr_2 p - \Wscr_2\Lscr_1p$ is in $C^\infty([0,T];C^4[0,L])$. Moreover, the function $z\in C([0,T];Z)$ determined by $w$ via the expression $z(t)=[w(\cdot,t) \ w_t(\cdot,t) \ w_t(0,t) \ w_{xt}(0,t)]$ for all $t\in[0,T]$ is the unique classical solution of \eqref{eq:beam1}-\eqref{eq:beam4} for the initial state $z_0 = z(0)$ and input $f(t)=w(L,t)$. \hfill$\square$
\end{remark}

Recall the operators $\Lscr_1$, $\Lscr_2$, $\Wscr_1$, $\Wscr_2$ from Proposition \ref{pr:psolution} and Remark \ref{rm:summary}. In the next theorem, building on the results in Propositions \ref{pr:genfnest}, \ref{pr:wisclasic} and \ref{pr:psolution}, we prove by construction the existence of a control input $f$ which transfers the beam model \eqref{eq:beam1}-\eqref{eq:beam4} between any two states belonging to a certain subspace of $Z$ over a prescribed time-interval. \vspace{1mm}

\begin{theorem}\label{th:main_result}
Fix a time $T>0$. Consider the set
\begin{align*}
 &\quad M = \Big\{[v(\cdot,0) \ v_t(\cdot,0) \ v_t(0,0) \ v_{xt}(0,0)] \in Z \m\Big |\m    \nonumber\\
  & v=\Wscr_1 \Lscr_2 p - \Wscr_2 \Lscr_1 p  \ \textrm{for some} \ p\in G_s[0,T], \ s\in (1,2) \Big\}.
\end{align*}
Let $z_0\in M$ and $z_T \in M$ be given. Then there exists an $f\in C([0,T];\rline)$ and a unique classical solution $z\in C([0,T];Z)$ of \eqref{eq:beam1}-\eqref{eq:beam4} on the time interval $[0,T]$ for the initial state $z_0$ and input $f$ such that $z(T)=z_T$.
\end{theorem}

\begin{proof}
From Remark \ref{rm:summary} it is evident that $M$ is a well-defined and non-empty set. Since $z_0,z_T\in M$ and $G_{s_1}[0,T] \subset G_{s_2}[0,T]$ for $s_1\leq s_2$, it follows from the definition of $M$  that there exist $p_0, p_T\in G_s[0,T]$ with $s\in (1,2)$ such that $v_0 = \Wscr_1 \Lscr_2 p_0 - \Wscr_2\Lscr_1 p_0$, $v_T = \Wscr_1 \Lscr_2 p_T - \Wscr_2\Lscr_1 p_T$,
$z_0= [v_0(\cdot,0) \ v_{0,t}(\cdot,0) \ v_{0,t}(0,0) \ v_{0,tx}(0,0)]$ and $z_T = [v_T(\cdot,0) \ v_{T,t}(\cdot,0) \ v_{T,t}(0,0) \ v_{T,tx}(0,0)]$.

For $t\in[0,T]$ let
\begin{equation}\label{eq:psi}
 \psi(t)=1-\Big(\int_0^t \psi_0(\tau)\dd \tau\bigg / \int_0^T \psi_0(\tau)\dd \tau\Big),
\end{equation}
where $\psi_0(t)=\exp\left(-\left[\left(1-\frac{t}{T} \right)\frac{t}{T} \right]^{-\frac{1}{s-1}} \right)$ for $t\in (0,T)$ and $\psi_0(0)=\psi_0(T)=0$. Then $\psi \in G_s[0,T]$ and
\begin{equation}\label{eq:psival}
 \psi(0) = 1, \quad \psi(T)=0, \quad  \psi^{(k)}(0)=\psi^{(k)}(T)=0
\end{equation}
for all $k\geq 1$, see \cite{MeScKu:2010}. For all $t\in [0,T]$ define
\begin{equation}\label{eq:pdef}
 p(t) = p_0(t)\psi(t)+p_T(T-t)\psi(T-t).
\end{equation}
Since $p_0, p_T,\psi \in G_s[0,T]$ and $G_s[0,T]$ is closed under addition and multiplication of functions \cite[Proposition 1.4.5]{Ro:1993} we get that $p\in G_s[0,T]$. Let $w = \Wscr_1 \Lscr_2 p-\Wscr_2 \Lscr_1 p$. Then $z(t) = [w(\cdot,t) \ w_t(\cdot,t) \  w_{xt}(\cdot,t) \ w_{xt}(\cdot,t)]\in C([0,T];Z)$ is the unique classical solution of \eqref{eq:beam1}-\eqref{eq:beam4} for the initial state $z(0)$ and input $f(t)=w(L,t)$, see Remark \ref{rm:summary}. We will now complete the proof of this theorem by establishing that $z(0)=z_0$ and $z(T)=z_T$.

The expression $w = \Wscr_1 \Lscr_2 p - \Wscr_2 \Lscr_1 p$ is the same as \eqref{eq:formalsoln} with $y_1 = \Lscr_2 p$ and $y_2 = - \Lscr_1 p$. The series for $w(\cdot,t)$, $w_t(\cdot,t)$, $w_t(0,t)$ and $w_{xt}(0,t)$ can be obtained by termwise differentiation of the series in \eqref{eq:formalsoln} (see the proof of Proposition \ref{pr:wisclasic}) so that $z(t)= [w(\cdot,t) \ w_t(\cdot,t)\ w_t(0,t) \ w_{xt}(0,t)]$ can be written as $\sum_{k\geq0} A_k y_1^{(k)}(t) + \sum_{k\geq0} B_k y_2^{(k)}(t)$ for $t\in[0,T]$. Here $A_k, B_k\in H^4(0,L)\times H^4(0,L)\times\rline\times\rline$ for $k\geq0$. Since $y_1^{(k)}(t) = [\Lscr_2 p^{(k)}](t)$ and $y_2^{(k)}(t) = -[\Lscr_1 p^{(k)}](t)$ (see the proof of Proposition \ref{pr:psolution}), using the definition of the operators $\Lscr_1$ and $\Lscr_2$ we get for $t\in[0,T]$,
\begin{equation} \label{eq:zdsum}
 z(t)=\sum_{k\geq0} \sum_{n\geq0} \bigl(A_k h_{n,x}(L)-B_k g_{n,x}(L)\bigr) p^{(n+k)}(t).
\end{equation}
Applying the above argument to $v_0 = \Wscr_1 \Lscr_2 p_0 - \Wscr_2 \Lscr_1 p_0$ and $v_T = \Wscr_1 \Lscr_2 p_T - \Wscr_2 \Lscr_1 p_T$ (instead of $w$) we get that $z_0=\sum_{k\geq0}  \sum_{n\geq0} \bigl( A_k h_{n,x}(L) - B_k g_{n,x}(L)\bigr) p_0^{(n+k)}(0)$, and $z_T=\sum_{k\geq0}  \sum_{n\geq0} \bigl( A_k h_{n,x}(L) - B_k g_{n,x}(L)\bigr) p_T^{(n+k)}(0)$. Differentiating \eqref{eq:pdef} $n$-times and then using \eqref{eq:psival} we get $p^{(n)}(0) = p_0^{(n)}(0)$ and $p^{(n)}(T) = p_T^{(n)}(0)$ for all $n\geq 0$. It now follows from the expressions for $z(0)$ and $z(T)$ from \eqref{eq:zdsum} and the expressions for $z_0$ and $z_T$ given above that $z(0)=z_0$ and $z(T)=z_T$. This completes the proof.
\end{proof}

\begin{remark} \label{rm:ss}
Let $c\in\rline$. For the initial state $[u_{ss} \ 0 \ 0 \ 0]\in Z$ with $u_{ss}(x)=c$ for all $x\in[0,L]$ and the constant input $f_{ss}\in C([0,T];\rline)$ with $f_{ss}(t)=c$ for all $t\in[0,T]$, note that the constant function $z_{ss}(t)=[u_{ss} \ 0 \ 0 \ 0]$ for $t\in [0,T]$ is the classical solution of \eqref{eq:beam1}-\eqref{eq:beam4}, see Definition \ref{def:classical}. We call $z_{ss}(0)$ the steady-state of \eqref{eq:beam1}-\eqref{eq:beam4} corresponding to the constant input $f_{ss}$. Each such steady-state is a rest configuration of the beam corresponding to some fixed position of the cantilever joint.

Let $v = \Wscr_1 \Lscr_2 p - \Wscr_2 \Lscr_1 p$ with $p(t)=c$ for all $t\in[0,T]$ and some $c\in\rline$. From the definition of the operators $\Lscr_1$, $\Lscr_2$, $\Wscr_1$, $\Wscr_2$ it follows that $[v(\cdot,0) \ v_t(\cdot,0) \ v_t(0,0) \ v_{xt}(0,0)]=[u_{ss} \ 0 \ 0 \ 0]$, where $u_{ss}(x)=c$ for all $x\in [0,L]$. In other words, steady-states of \eqref{eq:beam1}-\eqref{eq:beam4} belong to the set $M$. It now follows from Theorem \ref{th:main_result} that we can find an input $f$ to transfer  \eqref{eq:beam1}-\eqref{eq:beam4} between any two steady-states. \hfill$\square$
\end{remark}

\section{Numerical and experimental results} \label{sec5} 

Our experimental setup consists of a non-uniform moving cantilever beam, with linearly-varying width, made of stainless steel. One end of the beam supports a tip-mass, while the other end is attached via a cantilever joint to a cart mounted on a Hiwin single axis robot, see Figure 1. The robot is driven by a Yasaka AC servo motor which fixes the cart position as per the control input it receives from a Raspberry Pi 4B microprocessor. The dynamics of the beam is described by the 
model \eqref{eq:beam1}-\eqref{eq:beam4} with  $L=0.5\m\textrm{m}$, $m=0.402\m\textrm{kg}$, $J=1.9 \times 10^{-4} \m\textrm{kg\m m}^2$ and $\rho(x)=0.11(1+3x)\m\textrm{kg/m}$ and $EI(x)=0.297(1+3x)\m\textrm{N/m}$ for $x\in [0,L]$.

We consider two problems to illustrate our solution to the motion planning problem presented in Theorem \ref{th:main_result}.  In both the problems we take the time of transfer to be $T=3\m\textrm{s}$ and the final state to be $z_T =0$. In Problem 1 we take the initial state to be $z_0 = [v_0(\cdot,0) \ v_{0,t}(\cdot,0) \ v_{0,t}(0,0) \ v_{0,xt}(0,0)]$, where $v_0 = \Wscr_1\Lscr_2 p_0 - \Wscr_2 \Lscr_1 p_0$ with $p_0(t) = 1 + 10t^2e^{-2t}$ for $t\geq 0$. Note that $z_0\in M$ is not a steady-state of \eqref{eq:beam1}-\eqref{eq:beam4}. In Problem 2 we take the initial state to be the steady-state of \eqref{eq:beam1}-\eqref{eq:beam4} given by $z_0 = [0.4 \ 0 \ 0 \ 0]$, which is obtained from the above expressions for $z_0$ and $v_0$ by taking $p_0(t)=0.4$ for all $t\geq 0$, see Remark \ref{rm:ss}. We solve both the problems using the procedure described in the proof of Theorem \ref{th:main_result}. Accordingly we choose $\psi$ to be the function in \eqref{eq:psi} with $s=1.5$ and $T=3$ and define $p$ via \eqref{eq:pdef} by taking $p_T=0$ (since $z_T=0$).  The required control input is $f(t)=w(L,t)$, where $w=[\Wscr_1 \Lscr_2  - \Wscr_2 \Lscr_1]p$. 
Using the expressions for $\Lscr_1$, $\Lscr_2$, $\Wscr_1$, $\Wscr_2$ and changing the order of the double summations we have $f(t) = \lim_{N\to\infty} \sum_{l=0}^{N}\sum_{j+k=l} \big[g_k(L)h_{j,x}(L)+h_k(L)g_{j,x}(L) \big]p^{(2l)}(t)$. (Note that $g_k$ and $h_k$ are computed using the expressions in the proof of Proposition \ref{pr:genfnest}.) This series converges rapidly and by truncating it with $N=20$ we compute a very good approximation for the inputs which solve the two problems being considered, see Figure 2.
\vspace{-2mm}

$$\includegraphics{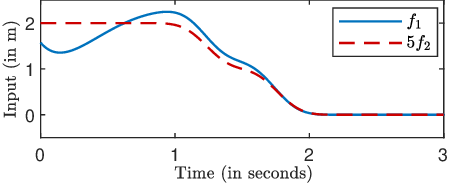}\vspace{-1mm}$$
\centerline{ \parbox{3.2in}{\small
Figure 2. Plot of the inputs $f_1$ and $f_2$ which solve Problem 1 and Problem 2, respectively. \vspace{1mm}}}
$$\includegraphics{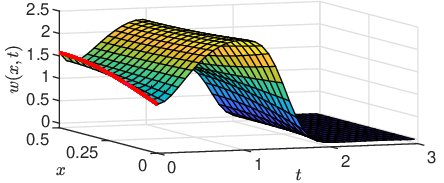}\vspace{-1mm}$$
\centerline{ \parbox{3.2in}{\small
Figure 3. Displacement profile $w(x,t)$ of the beam model \eqref{eq:beam1}-\eqref{eq:beam4} for the initial state and the input in Problem 1. The displacement profile starts from a non-constant (in $x$) function and settles down to the zero function in 3 seconds as expected.}}
\vspace{0.5mm}

We discretized the spatial derivatives in the beam model \eqref{eq:beam1}-\eqref{eq:beam4} using the finite-difference method (with step-size 1/300) to obtain a set of ODEs which serve as a numerical model for the beam model. We validated our solution for the motion planning problems, Problems 1 and 2, presented above by simulating the numerical model with the appropriate initial states $z_0$ and the control inputs shown in Figure 2. Figure 3 shows the beam displacement profile $w(x,t)$ obtained from our simulation for Problem 1. As expected (recall $z_T=0$), the displacement profile settles down to the zero function within 3 seconds. Figure 4 shows the displacement trajectory $w(0,t)$ of the tip-mass obtained from our simulation for Problem 2. As expected the tip-mass is initially at rest with $w(0,0)=0.4$ and again at rest finally with $w(0,3)=0$. We have implemented the control input in Problem 2 on our experimental setup in Figure 1 and observed that the beam is transferred from one steady-state to another (at a distance of 0.4\m{m}) within 3 seconds. We recorded the experiment using a camera at 240\m{fps} frame rate; The video of the experiment is available here:  {https://youtu.be/2IvgK5pK7Og}. By tracking the position of a green marker placed on the tip-mass using hue-based segregation and contour detection algorithms, we extracted the trajectory of the tip-mass from the video, see Figure 4. \vspace{-2mm}

$$\includegraphics{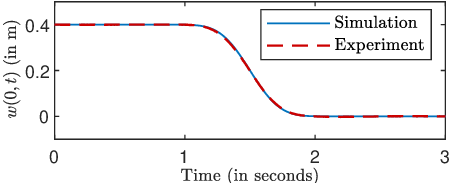}$$
\centerline{ \parbox{3.2in}{\small
Figure 4. Tip-mass position $w(0,t)$ obtained from the simulation and experiment for Problem 2 match closely. }}



\begin{thebibliography}{99}

\bibitem{AoFlMoRoRu:1997} Y. Aoustin, M. Fliess, H. Mounier, P. Rouchon and J. Rudolph, ``Theory and practice in the motion planning and control of a flexible robot arm using Mikusi\'nski operators,'' {\it Proc. 5th IFAC Symp. Robot Control}, pp. 267-273, Sep. 3-5, 1997, Nantes, France.

\bibitem{BaUlRu:2011} M. Bachmayer, H. Ulbrich and J. Rudolph, ``Flatness-based control of a horizontally moving erected beam with a point mass,'' {\it Math. Comput. Model. Dyn. Syst.,} vol. 17, pp. 49-69, 2011.


\bibitem{BaZh:2014} A. Badkoubeh and G. Zhu, ``A Green's function-based design for deformation control of a microbeam with in-domain actuation,'' {\it J. Dyn. Syst. Meas. Control}, vol. 136, 2014.

\bibitem{BaZhZh:2016} A. Badkoubeh, J. Zheng and G. Zhu, ``Flatness-based deformation control of an Euler-Bernoulli beam with in-domain actuation,'' {\it IET Control Theory Appl.}, vol. 10, pp. 2110-2118, 2016.


\bibitem{BaLy:2008} M. Barczyk and A. F. Lynch, ``Flatness-based estimated state feedback control for a rotating flexible beam: Experimental results,'' {\it IET Control Theory Appl.,} vol. 2, 288-302, 2008.










\bibitem{KoGeSc:2022} B. Kolar, N. Gehring and M. Sch\"oberl, ``On the calculation of differential parameterizations for the feedforward control of an Euler-Bernoulli beam,'' In {\it Dynamics and Control of Advanced Structures and Machines: Contributions from the 4th International Workshop, Linz,} pp. 123-136, Springer, 2022.






\bibitem{LyWa:2004} A. F. Lynch and D. Wang, ``Flatness-based control of a flexible beam in a gravitional field,'' {\it Proc. 2004 Amer. Control Conf.}, pp. 5449-5454, Jun. 30 - Jul. 2, 2004, Boston, USA.

\bibitem{MaRoRo:2016} P. Martin, L. Rosier and P. Rouchon, ``Null controllability of one-dimensional parabolic equations by the flatness approach,'' {\it SIAM J. Control Optim.}, vol. 54, pp. 198-220, 2016.


\bibitem{Me:2012} T. Meurer, {\it Control of Higher-Dimensional PDEs}, Springer-Verlag, Berlin, 2013.


\bibitem{MeThKu:2008} T. Meurer, D. Thull, and A. Kugi, ``Flatness based tracking control of a piezoactuated Euler-Bernoulli beam with non-collocated feedback: theory and experiments,'' {\it Int. J. Control}, vol. 81, pp. 475-493, 2008.

\bibitem{MeScKu:2010} T. Meurer, J. Schr\"ock, and A. Kugi, ``Motion planning for a damped Euler-Bernoulli Beam,'' {\it Proc. $49{\rm th}$ IEEE Conf. Decision \& Control}, pp. 2566-2571, Dec. 15-17, 2010, Atlanta, USA.


\bibitem{Pazy:1983} A. Pazy, {\it Semigroups of Linear Operators and Applications to Partial Differential Equations}, Springer-Verlag, New York, 1983.






\bibitem{RiTz:2009} G. G. Rigatos and S. G. Tzafestas, ``Flatness-based and energy-based control for distributed parameter robotic systems,'' {\it $13^{\rm th}$ IFAC Symp. Information Control Problems in Manufacturing,} pp.  1847-1852, Jun. 3-5, 2009, Moscow, Russia.


\bibitem{Ro:1993} L. Rodino, {\it Linear partial differential operators in Gevrey spaces}, World Scientific Publishing, Singpore, 1993.

\bibitem{ZhWe:2011} X. Zhao and G. Weiss, ``Controllability and observability of a well-posed system coupled with a finite-dimensional system,'' {\it IEEE Trans. Autom. Control,} vol. 56, pp. 88-99, 2011.

\end{thebibliography}
\end{document}